\documentclass[twocolumn]{IEEEtran} 

\usepackage{amsmath,amssymb,amsfonts}
\usepackage[final]{graphicx}
\usepackage{psfrag}
\usepackage[tight]{subfigure}
\usepackage[numbers,sort&compress]{natbib}
\usepackage{multicol}
\usepackage{mathrsfs} 
\usepackage{epstopdf}
\usepackage{multirow}
\usepackage{amsthm,mathtools}
\usepackage{float}
\usepackage{booktabs}  
\usepackage{multirow} 
\usepackage{relsize}
\usepackage{bm}
\usepackage{hyperref}
\usepackage{tabularx}
\usepackage[nolist,nohyperlinks]{acronym}
\usepackage{svg}
\usepackage{xparse}
\NewDocumentCommand{\evalat}{sO{\big}mm}{%
  \IfBooleanTF{#1}
   {\mleft. #3 \mright|_{#4}}
   {#3#2|_{#4}}%
}

\begin{acronym}
\acro{PLS}{physical layer security}
\acro{OP}{outage probability}
\acro{SNR}{signal-to-noise ratio}
\acrodefplural{SNR}[SNRs]{signal-to-noise ratios}
\acro{AWGN}{additive white Gaussian noise}
\acro{CSI}{channel state information}
\acro{PDF}{probability density function}
\acrodefplural{PDF}[PDFs]{probability density functions}
\acro{CDF}{cumulative distribution function}
\acro{PDF}{probability density function}
\acrodefplural{CDF}[CDFs]{cumulative distribution functions}
\acro{LRS}{large reflecting surface}
\acro{BS}{base station}
\acro{RF}{radio frequency}
\acro{RV}{random variable}
\acro{FN}{folded normal}
\acro{TDD}{time-division duplexing}
\acro{LOS}{line-of-sight}
\acro{DL}{downlink}
\acro{MRT}{maximal ratio transmission}
\acro{MRC}{maximal ratio combining}
\acro{MC}{Monte Carlo}
\acro{RV}{random variable}
\acro{CDF}{cumulative density function}
\acro{FAS}{fluid antenna system}
\acro{MIMO}{multipe-input multiple-output}
\acro{BS}{base station}
\acro{CSI}{channel state information}
\end{acronym}

\newtheorem{proposition}{Proposition}

\newtheorem{remark}{Remark}

\usepackage[numbers,sort&compress]{natbib}

\makeatletter
\def\blfootnote{\xdef\@thefnmark{}\@footnotetext}
\makeatother

\usepackage{float}
\usepackage{stfloats}
\usepackage{textcomp}
\usepackage[ruled,vlined]{algorithm2e}

\begin{document}
\title{\huge{Novel Expressions for the Outage Probability and Diversity Gains in Fluid Antenna System}}
\author{Jos\'e~David~Vega-S\'anchez, \textit{Member, IEEE}, Arianna Estefanía López-Ramírez, Luis~Urquiza-Aguiar, \textit{Member, IEEE}, and Diana~Pamela~Moya~Osorio, \textit{Senior Member, IEEE} }

\maketitle

\blfootnote{\noindent Manuscript received MONTH xx, YEAR; revised XXX. The review of this paper was coordinated by XXXX. The work of Luis Urquiza-Aguiar was supported by the Escuela Polit\'ecnica Nacional. The work of D. P. M. Osorio is supported by the Academy of Finland, project FAITH under Grant 334280.}

\blfootnote{\noindent
J.~D.~Vega-S\'anchez is with the Faculty of Engineering and Applied Sciences (FICA), Telecommunications Engineering, Universidad de Las Am\'ericas
(UDLA), Quito 170124, Ecuador (E-mail: $\rm jose.vega.sanchez@udla.edu.ec$), and also with the Departamento de Electr\'onica, Telecomunicaciones y Redes de Informaci\'on, Escuela Polit\'ecnica Nacional (EPN),
Quito,  170525, Ecuador. 
}
\blfootnote{\noindent 
A. E. López-Ramírez, and L.~Urquiza-Aguiar are with the  
Departamento de Electr\'onica, Telecomunicaciones y Redes de Informaci\'on, Escuela Polit\'ecnica Nacional (EPN), Ladr\'on de Guevara E11-253, Quito,  170525, Ecuador. (e-mail: cecilia.paredes@epn.edu.ec; luis.urquiza@epn.edu.ec).
}

\blfootnote{\noindent
D.~P.~Moya~Osorio is with the Centre for Wireless Communications (CWC), University of Oulu, Finland (e-mail: diana.moyaosorio@oulu.fi)
}

\vspace{-12.5mm}
\begin{abstract}
The flexibility and reconfigurability at the radio frequency (RF) front-end offered by the fluid antenna system (FAS) make this technology promising for providing remarkable diversity gains in networks with small and constrained devices. Toward this direction, this letter
compares the outage probability (OP) performance of
non-diversity and diversity FAS receivers undergoing spatially correlated Nakagami-$m$ fading channels. Although the system properties of FAS incur in complex analysis, we derive a simple yet accurate closed-form approximation by relying on a novel asymptotic matching method
 for the OP of a maximum-gain combining-FAS (MGC-FAS). The approximation is performed in two stages, the approximation of the cumulative density function (CDF) of each MGC-FAS branch, and then the approximation of the end-to-end CDF of the MGC-FAS scheme. With these results, closed-form expressions for the OP and the asymptotic OP are derived. Finally, numerical results validate our approximation of the MGC-FAS scheme and demonstrate its accuracy under different diversity FAS scenarios.
\end{abstract}

\begin{IEEEkeywords}
Asymptotic matching, maximum-gain combining-FAS (MGC-FAS), nakagami-$m$ fading, spatial correlation, outage probability.
\end{IEEEkeywords}

\vspace{-2.5mm}
\section{Introduction}
In recent years, \ac{MIMO} technology has been a fundamental part of the evolution of 5G and beyond to realize the impressive advancements in data rates and spectral efficiency. With \ac{MIMO}, diversity gain is guaranteed as long as the antennas are spatially separated by at least half wavelength. However, this may be challenging in very small devices of some Internet of Things (IoT) applications. Recently, a technology that uses liquid metals (e.g., gallium-indium eutectic, mercury, Galinstan) to design a software-controllable fluidic structure that, in its most basic implementation with only one \ac{RF} chain, allows a fluid radiator to switch among different positions in a small linear space, which has been referred to as a \ac{FAS}. In this way, \ac{FAS} can outperform traditional \ac{MIMO} regarding gains in diversity and multiplexing, specially when there exist space limitations at the receiver side \cite{Wong}.  

The performance of \ac{FAS} has been recently investigated in a number of works. For instance, in~\cite{Ref2}, Wong et al. introduced the novel concept of a single-antenna \ac{FAS} over correlated Rayleigh fading channels inspired by the advancement in mechanically flexible antennas. Afterward, in~\cite{Mukherjee}, Mukherjee et al. proposed a general framework for the evaluation of the second-order statistic  (i.e., the average level crossing rate) of the \ac{FAS} by considering time-varying fading channels. In~\cite{Yangyang}, Wong et al. revealed how the ergodic capacity scales with the system parameters of the \ac{FAS}.
In~\cite{Tlebaldiyeva1}, Tlebaldiyeva et al. derived a single-integral form of the \ac{OP} of a single-antenna \ac{FAS} over spatially correlated Nakagami-$m$ fading channels. A novel concept of fluid antenna multiple access (FAMA) was proposed in~\cite{AccessFAS}, which takes advantage of the deep fades suffered by the interference to attain a good channel condition  without demanding complex signal processing. In~\cite{Ref4}, Skouroumounis et al. presented an analytical framework based on stochastic geometry for evaluation the performance of large-scale FAS-aided cellular networks in terms of the \ac{OP}. In~\cite{New}, New et al. investigated the limit of \ac{FAS} performance via closed-form expressions of the OP and the diversity gain. In~\cite{Tlebaldiyeva2}, Tlebaldiyeva et al. recently compared non-diversity and diversity \ac{FAS} receivers undergoing $\alpha$-$\mu$ fading channels. Specifically, the diversity \ac{FAS} scheme considers enabling multiple ports of a fluid antenna and performing a combining technique with multi-port signals to further enhance \ac{FAS} performance further. Therein, a maximum-gain combining-\ac{FAS} (MGC-FAS) diversity scheme was investigated via Monte Carlo simulations due to the intricacy of the mathematical treatment for the underlying MGC-\ac{FAS}. In this sense, the OP and diversity gain for the MGC-\ac{FAS} are not known in closed-form expressions in the state-of-the-art. 

Motivated by the potential of the diversity FAS schemes to further enhance the capacity of future networks, with a great potential for IoT scenarios, we approximate the \ac{OP} and asymptotic OP for the MGC-FAS scheme in a closed-form fashion, which is useful for further evaluations of this scheme. For this purpose, we first approximate the \ac{CDF} of each MGC-FAS branch, and then, the \ac{CDF} of the MGC-FAS over correlated Nakagami-$m$ fading is derived. In both stages, the fitting parameters are estimated by employing the asymptotic matching method, proposed in \cite{Perim} that render a simple yet accurate approximation. To the best of the author's current knowledge, no prior work has provided a closed-form expression for the \ac{OP} of the MGC-FAS scheme in the literature. Additionally, in-depth insights about the impact of the antenna size on the \ac{OP} performance in diversity \ac{FAS} schemes are also provided.
\vspace{-2mm}

\section{System and Channel Models}
\begin{figure}[t]
\psfrag{A}[Bc][Bc][0.7]{$\mathrm{Transmitter}$}
\psfrag{B}[Bc][Bc][0.7]{$\mathrm{ports}$}
\psfrag{C}[Bc][Bc][0.7]{$1$}
\psfrag{D}[Bc][Bc][0.7]{$L$}
\psfrag{E}[Bc][Bc][0.6]{$W\times \lambda$}
\psfrag{G}[Bc][Bc][0.7][0.55]{$a)~M=1~ \mathrm{order}~\mathrm{FAS,}~\mathrm{Non}-\mathrm{diversity}$}
\psfrag{P}[Bc][Bc][0.7]{$b)~M- \mathrm{order}~\mathrm{MGC}-\mathrm{FAS}$}
\psfrag{L}[Bc][Bc][0.5]{$\frac{W}{M} \times \lambda$}
\psfrag{F}[Bc][Bc][0.7]{$g_{\mathrm{FAS}} =\max \left ( \left | g_1   \right |, \cdots, \left | g_L\right | \right )$}
\psfrag{W}[Bc][Bc][0.5]{$g_1^{\mathrm{FAS}} =\max \left ( \left | g_1   \right |, \cdots, \left | g_{L/M}\right | \right )$}
\psfrag{Z}[Bc][Bc][0.5]{$g_M^{\mathrm{FAS}} =\max \left ( \left | g_1   \right |, \cdots, \left | g_{L/M}\right | \right )$}
\psfrag{H}[Bc][Bc][0.7]{$g_{\mathrm{FAS}}^{\mathrm{MGC}}=\sum_{j=1}^{M}g^{\mathrm{FAS}}_j$}
\includegraphics[width=0.38\textwidth]{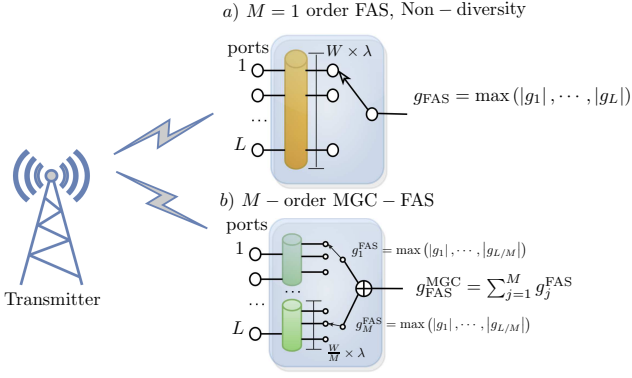}
\caption{System model for FAS-enabled communication with $a)$ MGC-FAS scheme and $b)$ non-diversity FAS configuration.
}
\label{SM}
\vspace{-4.5mm}
\end{figure}
Consider a point-to-point \ac{FAS} where the transmitter is equipped with a traditional antenna and the receiver with a fluid one with/without a diversity scheme, as described below.
\vspace{-6mm}
\subsection{Non-Diversity FAS Receiver}
In this scheme, the \ac{FAS}-receiver is built of a fluid antenna that can move freely along $L$-ports equally distributed along a linear dimension of length $W \lambda$, with $W$ being the antenna size and $\lambda$ the wavelength of the carrier frequency, as illustrated in Fig. \ref{SM}a. We assume that the \ac{FAS} can always switch to the best port for reaching the maximum received \ac{SNR}. Thus, the channel gain for the FAS can be expressed as
\vspace{-3mm}
\begin{equation}
\label{eq1}
g_{\mathrm{FAS}} =\max \left ( \left | g_1   \right |, \cdots, \left | g_L\right | \right ),
\end{equation}
where $g_{i}= \left | h_{i}\right |^{2}$ for $i\in\left \{ 1,\cdots,L \right \}$ denotes the channel gain of each port in the \ac{FAS}, and $h_i$ being modeled as a spatially-correlated Nakagami-$m$ fading channel because the antennas are located very close to each other in the linear space. The received SNR for non-diversity FAS receiver can be formulated as
\vspace{-3mm}
\begin{equation}\label{eq2}
\gamma=\frac{P \left | g_{\mathrm{FAS}}\right |}{N_0}= \overline{\gamma} \left | g_{\mathrm{FAS}}\right |, 
\end{equation}
where $\overline{\gamma}=\frac{P}{N_0}$ is the average transmit SNR, with $P$ being the transmit power and $N_0$ the additive
white Gaussian noise (AWGN) power.
\vspace{-4mm}
\subsection{Diversity FAS Receiver}
In this scheme, the entire FAS with $W\times \lambda$ is split into $M$ sub-FAS branches with a size of $W\times \lambda/M$, so the MGC-FAS equipment comprises $L/M$ ports per FAS branch, as shown in Fig. \ref{SM}b. In the MGC-FAS diversity scheme, the equivalent channel is the sum of the strongest channels of each FAS tube as
\vspace{-3mm}
\begin{equation}\label{eq3}
g_{\mathrm{FAS}}^{\mathrm{MGC}}=\sum_{j=1}^{M}g^{\mathrm{FAS}}_j,
\end{equation}
where $g_j^{\mathrm{FAS}} =\max \left ( \left | g_1   \right |, \cdots, \left | g_{L/M}\right | \right )$,
and $g_k=\left | h_{k}\right |^{2}$ for $k\in\left \{ 1,\cdots,L/M\right \}$ as in the non-diversity case, $h_{k}$ experiences correlated Nakagami-$m$ fading channel. The received \ac{SNR} for MGC-FAS receiver can be expressed as
\begin{equation}\label{eq4}
\gamma^{\mathrm{MGC}}=\frac{P \left | g_{\mathrm{FAS}}^{\mathrm{MGC}}\right | }{N_0}= \overline{\gamma} \left | g_{\mathrm{FAS}}^{\mathrm{MGC}}\right |.
\end{equation}
In the following sections an approximate statistical model for $g_{\mathrm{FAS}}^{\mathrm{MGC}}$ will be obtained, then the \ac{OP} distribution can be obtained .
\vspace{-5mm}
\section{Performance Analysis}
In this section, the \ac{OP} is considered to assess the performance of the FAS,  which is defined as the probability that the received SNR  is less than a threshold rate $\gamma_{th}$. 
\vspace{-5mm}
\subsection{Exact OP Distributions}
\textbf{Non-diversity case:} From~\cite[Eq.~(10)]{Tlebaldiyeva1}, the \ac{OP} for non-diversity FAS over correlated Nakagami-$m$ random variables (RVs) is given by
\vspace{-3mm}
\begin{align}\label{eq5}
P_{\mathrm{out}}(\gamma_{th})&= \frac{2^m m }{\Gamma(m)\Omega_1^{2m} }\int_{0}^{\sqrt{\frac{\gamma_{th}}{\overline{\gamma}}}}r_1^{2m}\exp\left ( -\frac{m r_1^2}{\Omega_1^{2}} \right )\nonumber \\ \times & \prod_{i=2}^{L}\left ( 1-Q_m\left ( \sqrt{\tfrac{2m\mu_i^2r_1^2}{\Omega_1^{2}\left ( 1-\mu_i^2 \right )}} \right ),\sqrt{\tfrac{2m\gamma_{th}}{\Omega_i^{2}\left ( 1-\mu_i^2 \right )\overline{\gamma}}} \right )dr_1,
\end{align}
where $\Gamma(\cdot)$ denotes the Gamma function, $Q_m(\cdot,\cdot)$ is the $m$-order Marcum Q-function, $m$ is the
 fading parameter, and $\Omega_k^2 $ indicates the average channel power of Nakagami-$m$ fading. Motivated by \cite{correlation}, we assume that the spatial correlation coefficient, denoted by, $\mu_i$, is given by
 \begin{equation}\label{eq6}
\mu^2 =\left | \frac{2}{L(L-1)}\sum_{i=1}^{L-1}(L-i)J_0\left ( \frac{2\pi iW}{L-1} \right )\right |, \hspace{1mm} \textit{for}\hspace{1mm} \mu=\mu_i~ \forall i, 
\end{equation}
where all the ports don't have a reference port or any port is a reference to any other port. Also, $J_0(\cdot)$
is the zero-order Bessel function of the first kind. 

\textbf{Diversity case:} Departing from \eqref{eq5} and applying the relationship  $F_{g_j^{\mathrm{FAS}}} (x_j)=P_{\mathrm{out}}(x_j \overline{\gamma} )$ for $j\in\left \{ 1,\cdots,M\right \}$, the \ac{CDF} of the equivalent channel gain in each sub-FAS branch of \eqref{eq3} is given by
\vspace{-3mm}
\begin{align}\label{eq7}
F_{g_j^{\mathrm{FAS}}} (x_j)&= \frac{2^m m }{\Gamma(m)\Omega_1^{2m} }\int_{0}^{\sqrt{x}}r_1^{2m}\exp\left ( -\frac{m r_1^2}{\Omega_1^{2}} \right ) \nonumber \\ \times \prod_{k=2}^{L/M} &\left ( 1-Q_m\left ( \sqrt{\tfrac{2m\mu_k^2r_1^2}{\Omega_1^{2}\left ( 1-\mu_k^2 \right )}} \right ),\sqrt{\tfrac{2mx_j}{\Omega_k^{2}\left ( 1-\mu_k^2 \right )}} \right )dr_1.
\end{align}
Then, the \ac{PDF} of \eqref{eq7}, i.e., $f_{g_j^{\mathrm{FAS}}} (x_j)$ can be obtained by computing the respective derivative.
Thus, by using Brennan’s approach~\cite{Brennan}, the end-to-end \ac{CDF} of the MGC-FAS in \eqref{eq3} is expressed as
\begin{align}\label{eqs8}
F_{g_{\mathrm{FAS}}^{\mathrm{MGC}}}(x)=&\int_{0}^{x}\int_{0}^{x-x_{M}}...\int_{0}^{x-\sum_{i=3}^{M}x_{i}} \int_{0}^{x-\sum_{i=2}^{M}x_{i}} \nonumber \\ \times
&f_{{{g_j^{\mathrm{FAS}}} },...,{g_{M}^{\mathrm{FAS}} }}(x_1,...,x_{M})dx_{1}dx_{2}...dx_{M-1}dx_{M},  
\end{align}
where $f_{{{g_j^{\mathrm{FAS}}} },...,{g_{M}^{\mathrm{FAS}} }}(x_1,...,x_{M})$ is the joint \ac{PDF} of the correlated branches.
%
Finally, the \ac{OP} for the MGC-FAS is computed as  $P^{\mathrm{MGC}}_{\mathrm{out}}(\gamma_{th}) = F_{g_{\mathrm{FAS}}^{\mathrm{MGC}}}\left ( \tfrac{\gamma_{th}}{\overline{\gamma}} \right )$. 
\subsection{Proposed OP Approximation}
It is noteworthy that the multi-fold integral in \eqref{eqs8} is quite intricate, thus the derivation of a closed-form solution appears to be unfeasible.
To overcome such limitation, an accurate approximation is proposed for the $P^{\mathrm{MGC}}_{\mathrm{out}}(\gamma_{th})$, which is obtained via the asymptotic matching method introduced in \cite{Perim}, as stated in the following proposition.
\begin{proposition}\label{Propos1}
The \ac{OP} expression of FAS undergoing correlated Nakagami-$m$ RVs can be approximated by
\begin{align}\label{eqs9}
P^{\mathrm{MGC}}_{\mathrm{out}}&\approx \frac{\Upsilon(\alpha_{\mathrm{MGC}} ,\frac{\gamma_{th}}{\beta_{\mathrm{MGC}}\overline{\gamma}})}{\Gamma{\left ( \alpha_{\mathrm{MGC}} \right )}},
\end{align}
where is $\Upsilon(\cdot,\cdot)$, is the lower incomplete gamma function~\cite[Eq.~(6.5.2)]{Abramowitz} and
\begin{subequations}
\vspace{-2mm}
\label{eqs10}
	\begin{align}
	\label{eq10a}
	\alpha_{\mathrm{MGC}}  =& \frac{L}{M}\sum_{j=1}^{M}m_j, \hspace{3mm} \beta_{\mathrm{MGC}}=\left (\frac{1}{\prod_{j=1}^{M}\frac{1}{(\beta_{j}^{\mathrm{FAS}} )^{\alpha_{j}^{\mathrm{FAS}}} }}  \right )^{1/\alpha_{\mathrm{MGC}} },
	\\
	\label{eq10b}
 \alpha_{j}^{\mathrm{FAS}}  =& \tfrac{L}{M}m_j, \hspace{3mm}
 \beta_{j}^{\mathrm{FAS}}=\left ( \frac{1}{\Gamma(\alpha_{j}^{\mathrm{FAS}} )a_{0,j}\alpha_{j}^{\mathrm{FAS}}} \right )^{1/\alpha_{j}^{\mathrm{FAS}} },
      \\
	\label{eq10c}
	a_{0,j} =&\frac{m_j^{m_j-1}}{\Gamma(m_j)\Omega_{1,j}^{2m_j} (m_j!)^{\tfrac{L}{M}-1}}\prod_{k=2}^{\tfrac{L}{M}}\left ( \frac{m_j}{\Omega_{k,j}^2\left ( 1-\mu_{k,j}^2 \right )} \right )^{m_j}.
	\end{align}
\end{subequations}
\end{proposition}
\begin{proof}
See Appendix~\ref{appendix1}.
\end{proof}
\begin{remark}\label{remark1}
Notice that the \ac{OP} expression is a novel and simple approximation that does not need to solve any involved integrals regarding the joint distribution of correlated fading channels of the MGC-FAS scheme.
\end{remark}
In order to attain more insights into the influence of system parameters
for the MGC-FAS performance,  an asymptotic closed-form expression for the \ac{OP} is derived. To this end, the asymptotic \ac{OP} is developed in the form $\mathrm{OP}^{\infty}\simeq \mathrm{G}_c\overline{\gamma}^{-\mathrm{G}_d}$~\cite{wang2003simple}, where $\mathrm{G}_c$ and $\mathrm{G}_d$ is the array gain and the diversity order, respectively. The asymptotic \ac{OP} is stated in the following Proposition.
\begin{proposition}\label{Propos2}
The asymptotic \ac{OP} expression for the proposed MGC-FAS over correlated Nakagami-$m$ RVs is given by
\begin{align}\label{eqs11}
P^{\mathrm{MGC}}_{\mathrm{out}}(\gamma_{th})&\simeq  \frac{(\frac{\gamma_{th}}{\beta_{\mathrm{MGC}}  \overline{\gamma}})^{\alpha_{\mathrm{MGC}} }}{\alpha_{\mathrm{MGC}}  \Gamma{\left (\alpha_{\mathrm{MGC}} \right )}},
\end{align}
\end{proposition}
\begin{proof}
By using $\Upsilon \left (a,x \right )\simeq x^a/a  $ as $x\rightarrow 0$ into~\eqref{eqs9},~\eqref{eqs11} is obtained straightforwardly.
\end{proof}
\begin{remark}\label{remark2}
\vspace{-2mm}
From~\eqref{eqs11} and~\eqref{eq10a}, it is clear that the diversity order reduces to $\mathrm{G}_d=Lm$ when all the sub-FAS tubes experience the same fading, i.e., $m=m_i,~\forall i $. Moreover, $\mathrm{G}_d$ is directly influenced by the number of ports and the fading severity.
\end{remark}

\section{Numerical results and discussions} \label{sect:numericals}
\begin{figure}[t]
\psfrag{A}[Bc][Bc][0.7]{$\mathrm{\textit{n}=49}$}
\psfrag{B}[Bc][Bc][0.7]{$\mathrm{\textit{n}=100}$}
\psfrag{C}[Bc][Bc][0.7]{$\mathrm{\textit{n}=196}$}
 \includegraphics[width=\linewidth]{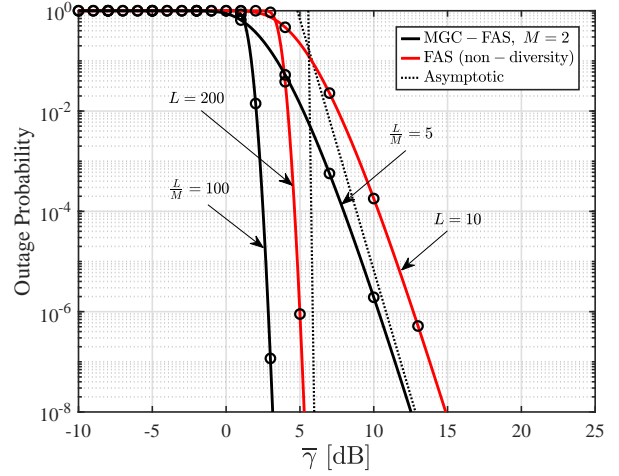}
\caption{\ac{OP} vs. $\overline{\gamma}$, for different numbers of ports by assuming $W=2$, $m=1$, and $\gamma_{th}=5$ dB. Markers denote the proposed approximation given in~\eqref{eqs9}, whereas the solid and dotted lines represent the analytical and the asymptotic solutions computed via~\eqref{eqs8} and~\eqref{eqs11}, respectively.}
\label{fig1}
\end{figure}
In this section, the impact of the system model parameters (e.g., the number of $M$ sub-FAS branches, the size of the antenna $W$, and the severity of fading) on the \ac{OP} performance is investigated, as well as the accuracy of the proposed approximations through illustrative examples. Unless stated otherwise, $\Omega_k=1, \forall k$ is considered for all plots, and the spatial correlation model is computed with the help of~\eqref{eq6}.  For the sake of comparison, $i)$ the traditional \ac{MRC} technique with uncorrelated antennas and $ii)$ the non-diversity FAS receiver, are included as a baseline in the \ac{OP} analysis.

In Fig~\ref{fig1}, we show the \ac{OP} versus the $\overline{\gamma}$ for $W=2$, $m=1$, $\gamma_{th}=5$ dB and by varying the number of ports $L=\left \{10,200 \right \}$ in the non-diversity FAS receiver. Hence, when setting  $M=2$-order MGC-FAS, it means that there are two sub-FAS tubes with $\tfrac{L}{M}=\left \{5,100 \right \}$ ports each. In this figure, the accuracy of the proposed approach in~\eqref{eqs9} to approximate the exact solution computed via~\eqref{eqs8} is evaluated. Note that all approximate curves are tight 
to the analytical solutions for the entire average \ac{SNR} range.
It is worth noting that as $M$ (i.e., the FAS tubes) increases, the exact formulation in~\eqref{eqs8} becomes computationally hard to compute, prone to convergence, or even unworkable. Hence, our simple and accurate approximation, with negligible computational cost, proves useful for the performance analysis for diversity FAS schemes. 
Considering the asymptotic curves, note that the \ac{OP} decline is steeper (i.e., good \ac{OP} performance) as the number of ports in the sub-FAS tubes or the severity fading $m$ increases. Contrariwise, the \ac{OP} is affected when the number of sub-FAS ports or the $m$ parameter decreases, so the \ac{OP} slope is  is less pronounced. These results are in coherence with the insights examined in Remark~\ref{remark2}. On the other hand, it is observed that the asymptotic \ac{OP} of the MGC-FAS quickly matches the diversity order of the exact solution for $\tfrac{L}{M}=5$. Conversely, for the scenario with $\tfrac{L}{M}=5$, the asymptotic \ac{OP} fits the correct asymptotic behavior for relatively lower operational \ac{OP} values. 

Henceforth, the approximation curves are represented in all plots with solid lines for visibility purposes. In Fig.~\ref{fig2},  the \ac{OP} is depicted as a function of the number of ports of the FAS scheme. For instance, the fixed value $L=100$ on the x-axis for a non-diversity FAS scheme corresponds to $\tfrac{L}{M}=50$ in the $2$-order MGC-FAS technique. The remaining parameters are set to: $\overline{\gamma}=1$ dB, $m=1$, and $\gamma_{th}=2$ dB. 
\begin{figure}[t]
\psfrag{A}[Bc][Bc][0.7]{$\mathrm{\textit{n}=49}$}
\psfrag{B}[Bc][Bc][0.7]{$\mathrm{\textit{n}=100}$}
\psfrag{C}[Bc][Bc][0.7]{$\mathrm{\textit{n}=196}$}
 \includegraphics[width=\linewidth]{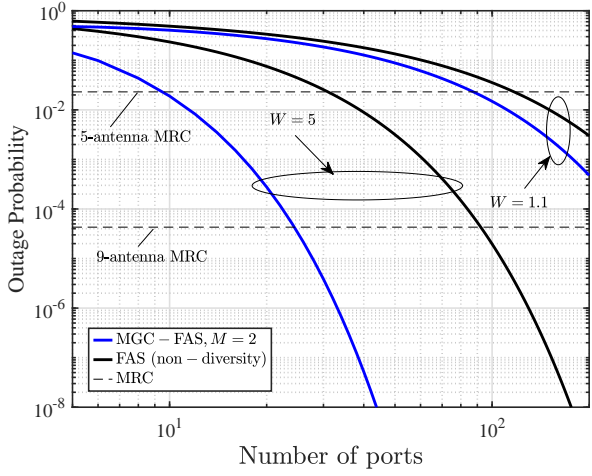}
\caption{\ac{OP} vs. number of ports by varying $W$ for $\overline{\gamma}=1$ dB, $m=1$, and $\gamma_{th}=2$ dB. Solid lines denote the proposed approximation in~\eqref{eqs9}.}
\label{fig2}
\vspace{-4.5mm}
\end{figure}
\begin{figure}[t]
\psfrag{A}[Bc][Bc][0.7]{$\mathrm{\textit{n}=49}$}
\psfrag{B}[Bc][Bc][0.7]{$\mathrm{\textit{n}=100}$}
\psfrag{C}[Bc][Bc][0.7]{$\mathrm{\textit{n}=196}$}
 \includegraphics[width=\linewidth]{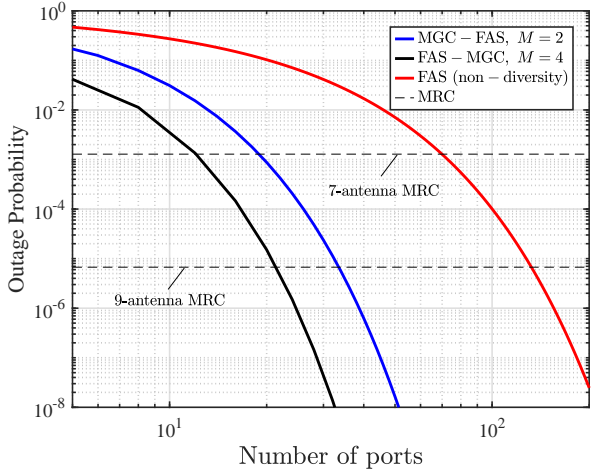}
\caption{\ac{OP} vs. number of ports with $M = \left \{ 2,4 \right \}$-order MGC-FAS for $W=3$, $m=1$, $\overline{\gamma}=1$ dB, and $\gamma_{th}=2$ dB. Solid lines denote the proposed approximation in~\eqref{eqs9}.}
\label{fig3}
\vspace{-4.5mm}
\end{figure}
In these results, the influence of changing the antenna size $W$ on the \ac{OP} behavior for MGC-FAS receivers is investigated. Note that large $W$ values (i.e., more space in the MGC-FAS) favor the \ac{OP} performance compared to the non-diversity FAS scheme. In particular, obtaining a higher performance gain from MGC-FAS over non-diversity FAS highly depends on the size coefficient $W$ for a fixed number of ports. Moreover, non-diversity/diversity FAS schemes beat the \ac{MRC} method. For instance, the \ac{OP} for the $9$-antenna \ac{MRC} is exceeded when the FAS is deployed by assuming $W=5$, $L=91$, and $\tfrac{L}{M}=24$ ports for non-diversity FAS and MGC-FAS schemes, respectively.
 
In Fig.~\ref{fig3}, the \ac{OP} is displayed as a function of the number of ports, as explained in Fig.~\ref{fig2}. Herein, the achievable \ac{OP} is examined by  comparing $M = \left \{ 2,4 \right \}$-order MGC-FAS with the non-diversity FAS receiver for $W=3$, $m=1$, $\overline{\gamma}=1$, and $\gamma_{th}=2$ dB. All plots show that increasing the sub-FAS tubes (i.e., $M$) greatly benefits the performance of the \ac{OP} compared to the non-diversity FAS receiver. In fact, this gain gap between $M$-order MGC-FAS and the FAS could be further boosted by increasing the size of the antenna $W$, as explained in Fig.~\ref{fig2}. Furthermore, $9$-antenna \ac{MRC} is surpassed when the FAS is assumed with  $\tfrac{L}{M}=20$, $\tfrac{L}{M}=34$ and $L=132$ ports for $4$-order MGC-FAS, $2$-order MGC-FAS, and non-diversity FAS schemes, respectively. This fact confirms the importance of using diversity in the FAS receiver.

Finally, Fig.~\ref{fig4} illustrates the OP versus $\gamma_{th}$ by varying the number of ports of non-diversity/diversity schemes for $W=2$, $\overline{\gamma}=0$ dB, $m=2$. Overall, it can be observed that increasing $\gamma_{th}$ leads to a noteworthy failure in the \ac{OP}. The best \ac{OP} performances are attained with a $2$-order MGC-FAS receiver. In addition, the $5$-antenna MRC is more quickly beaten using $2$-order MGC-FAS than its non-diversity FAS counterpart, as expected.

\begin{figure}[t]
\psfrag{A}[Bc][Bc][0.7]{$\mathrm{\textit{n}=49}$}
\psfrag{B}[Bc][Bc][0.7]{$\mathrm{\textit{n}=100}$}
\psfrag{C}[Bc][Bc][0.7]{$\mathrm{\textit{n}=196}$}
 \includegraphics[width=\linewidth]{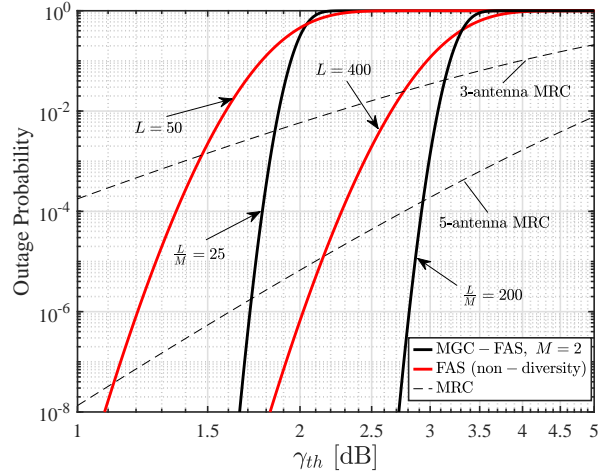}
\caption{\ac{OP} vs. $\gamma_{th}$ by varying the number of ports for $W=2$, $\overline{\gamma}=1$ dB, and $m=2$. Solid lines denote the proposed approximation in~\eqref{eqs9}.}
\label{fig4}
\vspace{-4.5mm}
\end{figure}
\vspace{-8mm}
\section{Conclusions}
In this letter, we examined the OP performance of a point-to-point FAS by
assuming non-diversity and diversity FAS receivers undergoing correlated Nakagami-$m$ fading channels. Specifically, a novel asymptotic matching method is employed to approximate the CDF of the MGC-FAS receiver in an analytically tractable way without incurring multi/single-fold integrals. With this by-product, a simple closed-form expression of the OP for the MGC-FAS scheme was derived. Furthermore, useful insights were provided concerning how the antenna size $W$ influences the \ac{OP} performance of the MGC-FAS. Specifically, the MGC-FAS scheme provided remarkable gains in terms of the OP over the non-diversity FAS when the $W$ values are large enough.
\appendices

\section{Proof of Proposition \ref{Propos1} }
\label{appendix1} 
In the first stage, a suitable approximation for the \ac{CDF} of each sub-FAS branch given in \eqref{eq7} is obtained by using a Gamma distribution. Hence, we employ the asymptotic matching method. Toward that~\cite[Eq.~(3)]{Xianchang} is replaced into~\eqref{eq7}, which is re-expressed as
\begin{align}\label{Apeq1}
F_{g_j^{\mathrm{FAS}}} (x_j)&\approx \frac{2{m{_j}} ^{m_j} }{\Gamma(m_j)\Omega_{1,j}^{2m_j} }\underset{I_1}{\underbrace{\int_{0}^{\sqrt{x_j}}r_1^{2m_j-1}\exp\left ( -\frac{m_j r_1^2}{\Omega_{1,j}^{2}} \right )}}\nonumber \\ \times & \underset{I_1}{\underbrace{\prod_{k=2}^{\tfrac{L}{M}}\left ( \tfrac{\left ( \tfrac{m_j x_j}{\Omega_{k,j}^{2}\left ( 1-\mu_{k,j}^2 \right )}  \right )^{m_j} \exp\left (-\tfrac{m_j\mu_{k,j}^2r_1^2}{\Omega_{1,j}^{2}\left ( 1-\mu_{k,j}^2 \right )}  \right )}{m_j!} \right )dr_1}}.
\end{align}
Here, with the aid of~\cite[Eq.~(3.381.1)]{Gradshteyn},~$I_1$ can be evaluated in closed-fashion in terms of the incomplete Gamma function, i.e.,  $\Upsilon(\cdot,\cdot)$. Then, by applying the relationship $\Upsilon \left (a,x \right )\simeq x^a/a  $ as $x\rightarrow 0$,  the asymptotic behavior of the \ac{CDF} of each sub-FAS branch in the form $F_{g_j^{\mathrm{FAS}}} (x_j)\simeq a_0 x_j^{b_0}$, can be formulated as
\vspace{-3mm}
\begin{align}\label{Apeq2}
\vspace{-2mm}
F_{g_j^{\mathrm{FAS}}} (x_j)&\simeq\underset{a_0}{\underbrace{\tfrac{m_j^{m_j-1}  m_j!^{1-\tfrac{L}{M}} }{\Gamma(m_j)(\Omega_{1,j})^{2m_j} }\prod_{k=2}^{\tfrac{L}{M}}\left ( \tfrac{m_j}{\Omega_{k,j}^2(1-\mu_{k,j}^2)} \right )^{m_j}}} x_j^{  \overset{b_0}{\overbrace{m_j\tfrac{L}{M} }} }.
\end{align}
To find the shape parameters of the Gamma distribution to approximate \eqref{eq7}, the asymptotic Gamma \ac{CDF}\footnote{To asymptotically approximate~$ F_{g_j^{\mathrm{FAS}}} (x_j)\approx \tfrac{\Upsilon\left ( \alpha_{j}^{\mathrm{FAS}} , \frac{x_j}{\beta_{j}^{\mathrm{FAS}} } \right )}{\Gamma{(\alpha_{j}^{\mathrm{FAS}} )}}$, the relationship $\Upsilon \left (a,x \right )\simeq x^a/a  $ as $x\rightarrow 0$, is employed.}
of each sub-FAS branch  is used to obtain the following expression
\vspace{-5mm}
        \begin{align}
         \label{Apeq3}
     \widetilde{F}_{g_j^{\mathrm{FAS}}} (x_j)\simeq  \underset{\widetilde{a}_0}{\underbrace{\tfrac{1}{{\beta_{j}^{\mathrm{FAS}}}^{\alpha_{j}^{\mathrm{FAS}}} \alpha_{j}^{\mathrm{FAS}} \Gamma ({\alpha_{j}^{\mathrm{FAS}} } )}}} x_j^{\overset{\widetilde{b}_0}{\overbrace{\alpha_{j}^{\mathrm{FAS}} }}}.
     \vspace{-2mm}
         \end{align}
Then, by applying the asymptotic matching~\cite{Parente}, i.e., $a_0=\widetilde{a}_0$ and $b_0=\widetilde{b}_0$, the shape parameters $\alpha_{j}^{\mathrm{FAS}} $ and $\beta_{j}^{\mathrm{FAS}}$ of the Gamma distribution to approximate \eqref{eq7} can be expressed as~\eqref{eq10b}. In the second stage, we approximate the \ac{CDF} of the MGC-FAS in \eqref{eqs8} by using again the Gamma distribution via the asymptotic matching technique. For this purpose, the approximate \ac{PDF} and \ac{CDF} of the MGC-FAS can be expressed as
\vspace{-3mm}
\begin{subequations}
\label{Apeq4}
	\begin{align}
	\label{Apeq4a}	\widetilde{f}_{g_{\mathrm{FAS}}^{\mathrm{MGC}}}(x)=&\tfrac{x^{\overset{\widetilde{d}_0}{\overbrace{\alpha_{\mathrm{MGC}}-1}}}
 }{\underset{\widetilde{c}_0}{\underbrace{\Gamma\left ( \alpha_{\mathrm{MGC}} \right ){\beta_{\mathrm{MGC}}}^{ \alpha_{\mathrm{MGC}} } }}}\exp\left ( -\tfrac{x}{\beta_{\mathrm{MGC}}
} \right )
	\\
	\label{Apeq4b}	\widetilde{F}_{g_{\mathrm{FAS}}^{\mathrm{MGC}}}(x)= & \tfrac{\Upsilon(\alpha_{\mathrm{MGC}} ,\tfrac{x}{\beta_{\mathrm{MGC}}})}{\Gamma{\left ( \alpha_{\mathrm{MGC}} \right )}},
	\end{align}
\end{subequations}
where $\widetilde{c}_0$ and $\widetilde{d}_0$ are the linear and the angular coefficients that capture the asymptotic behavior of the approximate distribution, i.e., $\widetilde{f}_{g_{\mathrm{FAS}}^{\mathrm{MGC}}}(x)$. Now, we are interested in the asymptotic behavior of the PDF of the sum $f_{g_{\mathrm{FAS}}^{\mathrm{MGC}}}(x)$ given in \eqref{eq3}. Hence, by appropriately substituting the shape parameters of the summands, i.e., $\alpha_{j}^{\mathrm{FAS}} $ and $\beta_{j}^{\mathrm{FAS}}$ into~\cite[Eq.~(4)]{Parente} and after some manipulations, the linear and angular coefficients that govern the asymptote of the sum, $f_{g_{\mathrm{FAS}}^{\mathrm{MGC}}}(x)\simeq c_0 x^{d_0}$, are given by
\vspace{-3mm}
\begin{align}\label{Apeq5}
 c_0 = \frac{\prod_{j=1}^{M}\left [ \tfrac{1}{\Gamma\left ( \alpha_{j}^{\mathrm{FAS}} \right ){\beta_{j}^{\mathrm{FAS}}}^{\alpha_{j}^{\mathrm{FAS}}} }  \right ]}{\Gamma \left( \sum_{j=1}^{M}\alpha_{j}^{\mathrm{FAS}} \right )} , \hspace{2mm}  d_0  = -1+\sum_{j=1}^{M}\alpha_{j}^{\mathrm{FAS}}.
\end{align}
Next, by matching $c_0=\widetilde{c}_0$ and $d_0=\widetilde{d}_0$, the fitting parameters of \eqref{Apeq4} are found straightforwardly. Finally,~\eqref{eqs9} is obtained with the help of~\eqref{Apeq4b} by setting $P^{\mathrm{MGC}}_{\mathrm{out}}(\gamma_{th}) \approx \widetilde{F}_{g_{\mathrm{FAS}}^{\mathrm{MGC}}}\left ( \tfrac{\gamma_{th}}{\overline{\gamma}} \right )$. This completes the proof.

\bibliographystyle{ieeetr}
\bibliography{bibfile}

\end{document}